\let\oldvec\vec

\documentclass[ruled,vlined,linesnumbered,commentsnumbered]{llncs}

\let\vec\oldvec

\usepackage[ruled,vlined,linesnumbered,commentsnumbered]{algorithm2e}
\usepackage{amsmath,amssymb,amsfonts}
\usepackage{graphicx}

\newcommand {\ignore} [1] {}

\begin{document}
\newtheorem{fact}[lemma]{Fact}

\title{An $\tilde{O}(\log^2 n)$-approximation algorithm for $2$-edge-connected dominating set}
\author{Amir Belgi\thanks{Part of this work was done as a part of author's M.Sc. Thesis at the Open University of Israel.}
\and Zeev Nutov} 

\institute{The Open University of Israel. \email{belgiamir251@gmail.com, nutov@openu.ac.il}}

\maketitle


\def\c     {\sc Connected}
\def\ec   {\sc Edge-Connected}

\def\ds   {\sc Dominating Set}
\def\dsu {\sc Dominating Subgraph}

\def\dt {\sc Dominating Subtree}

\def\au {\sc Augmentation}
\def\st  {\sc Steiner}
\def\su  {\sc Subset}
\def\nw  {\sc Node Weighted}


\def\si     {\sigma}
\def\eps  {\epsilon}
\def\al     {\alpha}

\def\empt {\emptyset}
\def\sem  {\setminus}
\def\subs  {\subseteq}

\def\t   {\tilde}
\def\h   {\hat}

\def\f   {\frac}
\def\opt {\mathsf{opt}}

\begin{abstract}
In the {\c} {\ds} problem we are given a graph $G=(V,E)$ and seek a minimum size dominating set $S \subs V$ 
such that the subgraph $G[S]$ of $G$ induced by $S$ is connected.
In the $2$-{\ec} {\ds} problem $G[S]$ should be $2$-edge-connected.  
We give the first non-trivial approximation algorithm for this problem, with expected approximation ratio $\t{O}(\log^2n)$. 
\end{abstract}

\section{Introduction} \label{s:intro}

Let $G=(V,E)$ be a graph.
A subset $S \subs V$ of nodes of $G$ is a {\bf dominating set} in $G$ if every $v\in V\sem S$ has a neighbors in $S$.
In the {\ds} problem the goal is to find a min-size dominating set $S$.
In the {\c} {\ds} problem the subgraph $G[S]$ of $G$ induced by $S$ should be connected.
This problem admits a tight approximation ratio $O(\log n)$, even in the node weighted case \cite{GK98,GK99}, 
based on \cite{KR}. 

A graph is {\bf $2$-edge-connected} if it contains $2$ edge disjoint paths between every pair of nodes.
We consider the following problem.

\medskip

\begin{center} \fbox{\begin{minipage}{0.97\textwidth}
\underline{$2$-{\ec} {\ds}} \\
{\em Input:} \ \ A graph $G=(V,E)$. \\
{\em Output:} A  min-size dominating set $S \subs V$ such that $G[S]$ is $2$-edge-connected.
\end{minipage}} \end{center}

\medskip

Given a distribution ${\cal T}$ over spanning trees of a graph $G$, the {\bf stretch} of ${\cal T}$ is
${\displaystyle \max_{uv \in E}} \mathbb{E}_{T \sim {\cal T}} \left[\f{d_T (u,v)}{d_G(u,v)}\right]$, where $d_H(u,v)$ denotes the 
distance between $u$ and $v$ in a graph $H$.
Let $\si=\si(n)$ denote the lowest known upper bound on the stretch that can be achieved 
by a polynomial time construction of such ${\cal T}$ for a graph on $n$ nodes. 
By the work of Abraham, Bartal, and Neiman \cite{ABN}, 
that is in turn based on the  work of Elkin, Emek, Spielman, and Teng \cite{EEST}, we have:
$$
\si(n)=O(\log n \cdot\log\log n\cdot(\log\log\log n)^{3})=\t{O}(\log n) \ .
$$

Our main result is:

\begin{theorem} \label{t:main}
$2$-{\ec} {\ds} admits an approximation algorithm with expected approximation ratio $O(\si \log n)=\t{O}(\log^2 n)$.
\end{theorem}

In the rest of the Introduction we discuss motivation, related problems, and give a 
road-map of the proof of Theorem~\ref{t:main}. 

It is a common problem in network design to route messages through the network. 
Many routing protocols exploit flooding strategy in which every node broadcasts the message to all of its neighbors. 
However, such protocols suffer from a large amount of redundancy. 
Ephremides, Wieselthier, and Baker \cite{E87} introduced the idea of constructing a {\bf virtual backbone} of a network.
A virtual backbone is often chosen to be a {\bf connected dominating set} --  a connected subgraph (a tree) on a dominating node set.
Then only the nodes of the tree are involved in the routing,
which may significantly reduce the number of messages the routing protocol generates. 
Moreover, we only need to maintain the nodes in the tree to keep the message flow. 
This raises the natural problem of constructing a ``cheap'' connected virtual backbone $H$.
Usually ``cheap'' means that $H$ should have a minimum number of edges or nodes,
or, more generally, that we are given edge costs/node weights, and $H$ should have a minimum cost/weight. 

In many cases we also require from the virtual backbone to be robust to edge or node failures.
A graph $G$ is {\bf $k$-edge-connected} if it contains $k$ edge disjoint paths between every pair of nodes; 
if the paths are required to be internally node disjoint then $G$ is {\bf $k$-connected}.
A subset $S$ of nodes in a graph $G=(V,E)$ is an {\bf $m$-dominating set} if every $v \in V \sem S$ has at least $m$ neighbors in $S$.
In the {\sc Min-Weight $k$-Connected $m$-Dominating Set} problem we seek a minimum node weight $m$-dominating set $S$
such that the subgraph $G[S]$ of $G$ induced by $S$ is $k$-connected. 
This problem was studied in many papers, both in general graphs and in unit disk graph, 
for arbitrary weights and also for unit weights; the unit weights case is the {\sc $k$-Connected $m$-Dominating Set} problem. 
We refer the reader to recent papers \cite{F,ZZMD,N-CDS}.
In the {\sc Min-Cost $k$-Connected $m$-Dominating Subgraph} problem, we seek to minimize 
the cost of the edges of the subgraph rather than the weight of the nodes. 
We observe that for unit weights/costs, the approximability of the {\sc $k$-Connected $m$-Dominating Set} problem 
is equivalent to the one of the {\sc $k$-Connected $m$-Dominating Subgraph} problem, up to a factor of $2$;
this is so since the number of edges in a minimally $k$-connected graph is between $kn/2$ and $kn$.  
The same holds also for the $k$-edge-connectivity variant of these problems.

Most of the work on the {\sc Min-Weight $k$-Connected $m$-Dominating Set} problem focused 
on the easier case $m \geq k$, when the union of a partial solution and a feasible solution is
always feasible. This enables to construct the solution by computing first an $\al$-approximate $m$-dominating set
and then a $\beta$-approximate augmenting set to satisfy the connectivity requirements;
the approximation ratio is then bounded by the sum $\al+\beta$ of the ratios of the two sub-problems.
The currently best ratios when $m\geq k$ are \cite{N-CDS}:
$O(k \ln n)$ for general graphs, $\min\left\{\f{m}{m-k},k^{2/3}\right\} \cdot O(\ln^2 k)$ for unit disc graphs, 
and $\min\left\{\f{m}{m-k},\sqrt{k}\right\} \cdot O(\ln^2 k)$ for unit disc graphs with unit node weights.
However, when $m<k$ this approach does not work, and the only non trivial ratio known is for (unweighted) unit disk graphs,
due to Wang et al. \cite{WKAG}, where they obtained a constant ratio for $k\leq 3$ and $m=1,2$.
It is an open question to obtain a non-trivial ratio for the (unweighted) 
{\sc $2$-Connected Dominating Set} problem in general graphs. 

The $2$-{\ec} {\ds} problem that we consider is the edge-connectivity version of the above problem, 
when the virtual backbone should be robust to single edge failures.
As was mentioned, the approximability of this problem is the same, up to a factor of $2$, 
as that of the $2$-{\ec} {\dsu} problem that seeks to minimize 
the number of edges of the subgraph rather than the number of nodes.
We prove Theorem~\ref{t:main} for the latter problem using a two stage reduction. 
Our overall approximation ratio $O(\si \log n)$ is a product of the first reduction fee $\si$ and 
the approximation ratio $O(\log n)$ for the problem obtained from the second reduction.

In the first stage (see Section~\ref{s:red}) 
we use the probabilistic embedding into a spanning tree of \cite{ABN} with stretch $\si=\tilde{O}(\log n)$
to reduce the problem to a ``domination version'' of the so called {\sc Tree Augmentation} problem (c.f. \cite{EFKN}); 
in our problem, which we call {\sc Dominating Subtree}, 
we are given a spanning tree $T$ in $G$ and seek a min-size edge set $F \subs E \sem T$ 
and a subtree $T'$ of $T$, such that $T'$ dominates all nodes in $G$ and $T' \cup F$ is $2$-edge-connected.
This reduction invokes a factor of $\si=\tilde{O}(\log n)$ in the approximation ratio. 
Gupta, Krishnaswamy, and Ravi \cite{GKR2} used such tree embedding to give a generic framework 
for approximating various restricted $2$-edge-connected network design problems, 
among them the {\sc $2$-Edge-Connected Group Steiner Tree} problem.
However, all their algorithms are based on rounding an appropriate LP relaxations, while our algorithm is purely combinatorial
and uses different methods. 

In the second stage (see Section~\ref{s:red'}) we reduce the {\sc Dominating Subtree} problem to  
the {\sc Subset Steiner Connected Dominating Set} problem \cite{GK98}.
While we show in Section~\ref{s:hard} that in general this problem is as hard as the {\sc Group Steiner Tree} problem,
the instances that are derived from the reduction have special properties that will enable us to obtain ratio $O(\log n)$. 
We note that the reduction we use is related to the one of Basavaraju et al. \cite{BFGM},
that showed a relation between the {\sc Tree Augmentation} and the {\sc Steiner Tree} problems.

\section{Reduction to the dominating subtree problem} \label{s:red}

To prove Theorem~\ref{t:main} we will consider the following variant of our problem:

\begin{center} \fbox{\begin{minipage}{0.97\textwidth}
\underline{$2$-{\ec} {\dsu}} \\
{\em Input:} \ \ A graph $G=(V,E)$. \\
{\em Output:} A $2$-edge-connected subgraph $(S,J)$
of $G$ with $|J|$ minimum such that $S$ is a dominating set in $G$. 
\end{minipage}} \end{center}

Since $|S| \leq |J| \leq 2(|S|-1)$ holds for any edge-minimal $2$-edge-connected graph $(S,J)$, 
then if $2$-{\ec} {\dsu} admits ratio $\rho$ then $2$-{\ec} {\ds} admits ratio $2\rho$.
Thus it is sufficient to prove Theorem~\ref{t:main} for the $2$-{\ec} {\dsu} problem.

For simplicity of exposition we will assume that we are given a single spanning tree $T=(V,E_T)$ with stretch $\si$, namely that 
$$
|T_f| \leq \si  \ \ \ \ \forall f \in E \sem E_T
$$
where $T_f$ denotes the path in the tree $T$ between the endnodes of $f$.
We say that {\bf $f \in E \sem E_T$ covers $e \in E_T$} if $e \in T_f$.
For an edge set $F$ let $T_F=\cup_{f \in F} T_f$ denote the forest formed by the tree edges of $T$
that are covered by the edges of $F$.
The following two lemmas give some cases when $T_F \cup F$ is a $2$-edge-connected graph.

\begin{lemma} \label{l:feasible}
If $F \subs E \sem E_T$ then $T_F \cup F$ is $2$-edge-connected if and only if $T_F$ is a tree.
\end{lemma}
\begin{proof}
Note that every $f \in F$ has both ends in $T_F$.
It is known (c.f. \cite{EFKN}) that if $T'$ is a tree and $F$ is an additional edge set on the node set of $T'$, 
then $T' \cup F$ is $2$-edge-connected if and only if $T'_F=T'$.
This implies that if $T_F$ is a tree then $T_F \cup F$ is $2$-edge-connected.
Now suppose that $T_F$ is not a tree. Let $C$ be a connected component of $T_F$. 
Then no edge in $F \cup T_F$ has exactly one end in $C$. Thus $C$ is also a connected component of $T_F \cup F$,
so $T_F \cup F$ is not connected.
\qed
\end{proof}

\begin{lemma} \label{l:si-loss}
If $(S,J)$ is a $2$-edge-connected subgraph of $G$ then $T_{J \sem E_T}$ is a tree.
\end{lemma}
\begin{proof}
By Lemma~\ref{l:feasible} the statement is equivalent to claiming that $T_{J \sem E_T} \cup (J \sem E_T)$ is $2$-edge-connected.
To see this, note that $T_{J \sem E_T} \cup (J \sem E_T)$ is obtained from the $2$-edge-connected graph $(S,J)$
by sequentially adding for each $f \in J \sem E_T$ the path $T_f$.
It is known that adding a simple path $P$ between two nodes of a $2$-edge-connected graph  
results in a $2$-edge-connected; this is so also if $P$ contains some edges of the graph. 
The statement now follows by induction.
\qed
\end{proof}

Let us consider the following problem. 

\medskip

\begin{center} \fbox{\begin{minipage}{0.98\textwidth}
\underline{{\dt}} \\
{\em Input:} \ \ A graph $G=(V,E)$ and a spanning tree $T=(V,E_T)$ in $G$.  \\
{\em Output:} A min-size edge set $F \subs E \sem E_T$ such that $T_F$ is a dominating tree.
\end{minipage}} \end{center}

\medskip

From Lemmas \ref{l:feasible} and \ref{l:si-loss} we have the following.

\begin{corollary}
Let $(S,J)$ be an optimal solution of a $2$-{\ec} {\dsu} instance $G$.
Let $T$ be a spanning tree in $G$ with stretch $\si$ and $F$ a $\rho$-approximate solution to the {\dt} instance $G,T$.
Then $T_F \cup F$ is a feasible solution to the $2$-{\ec} {\dsu} instance and $|F \cup E(T_F)| \leq \rho(\si+1)|J|$.
\end{corollary}
\begin{proof}
$T_F \cup F$ is a feasible solution to the $2$-{\ec} {\dsu} instance by 
the definition of the {\dt} problem and Lemma~\ref{l:feasible}. 
By Lemma~\ref{l:si-loss}, $J \sem E_T$ is a feasible solution to the {\dt} instance,
thus $|F| \leq \rho |J \sem E_T| \leq \rho|J|$.
Since $T$ has stretch $\si$ we get $|E(T_F)| \leq \si |F| \leq \si \rho |J|$.
\qed
\end{proof}

Hence to finish the proof of Theorem~\ref{t:main} is is sufficient to prove the following  theorem, 
that may be of independent interest.

\begin{theorem} \label{t:main'}
The {\dt} problem admits approximation ratio $O(\log n)$.
\end{theorem}

\section{Reduction to subset connected dominating set} \label{s:red'}

In this section we reduce the {\dt} problem to the {\su} {\st} {\c} {\ds}, 
and show that the special instances that arise from the reduction admit ratio $O(\log n)$.
The justification of the reduction is given in the following lemma.

\begin{lemma} \label{l:sv}
Let $T=(V,E_T)$ be a tree and $F$ and edge set on $V$, and let $s,t \in V$. 
Let $H=(F \cup V,I)$ be a bipartite graph where $I=\{fv: f \in F, v \in V \cap T_f\}$.
Then $T \cup F$ has $2$ edge disjoint $st$-paths if and only if $H$ has an $st$-path.
\end{lemma}
\begin{proof}
Let $P=T_{st}$ be the $st$-path in $T$. 
By Menger's Theorem, $T \cup F$ has $2$ edge disjoint $st$-paths if and only if for every $e \in P$ there is $f \in F$ that covers $e$.
Let $S=\{v \in P:H \mbox{ has an } sv\mbox{-path}\}$. Let $\hat{P}$ be the set of edges in $P$ uncovered by $F$.
We need to show that $\hat{P} \neq \empt$ if and only if $t \notin S$.

Suppose that $t \notin S$. 
Among the nodes in $S$, let $u$ be the furthest from $s$ along $P$.
Let $v$ be the node in $P$ after $u$. Then $uv \in P$, since $u \neq t$. 
We claim that $uv \in \hat{P}$. Otherwise, there is $f \in F$ with $u,v \in T_f$ and we get that $v \in S$ (since $u \in S$), 
contradicting the choice of $u$.

Suppose that there is $e \in \hat{P}$. Let $T_s,T_t$ be the two trees of $T \sem \{e\}$, where $s \in T_s$ and  $t \in T_t$. 
Since no link in $F$ covers $e$, every link in $F$ has both ends either in $T_s$ or in $T_t$.
This implies that no node in $T_t$ belongs to $S$.
\qed
\end{proof}

Let us assume w.l.o.g. that our {\dt} instance consists of a tree $T$ on $V$ and 
an edge set $E$ on $V$ that contains no edge from $T$. 

\begin{definition} \label{d:con-dom}
Given a {\dt} instance $T,E$ the {\bf connectivity-domination graph} $\hat{G}=(\hat{V},\hat{E})$ 
has node set $\hat{V}=E \cup V$ and edge set $\hat{E}=I \cup D$ where (see Fig.~\ref{f:red}):
\begin{eqnarray*}
I  & = & \{fg: f,g \in E, V(T_f) \cap V(T_g) \neq \empt\} \\ 
D & =& \{ev: e \in E, v \in V, v \in T_e \mbox{ or } v \mbox{ is a neighbor of } T_e \mbox{ in } G\}
\end{eqnarray*}
\end{definition}

\begin{figure}
\centering 
\includegraphics{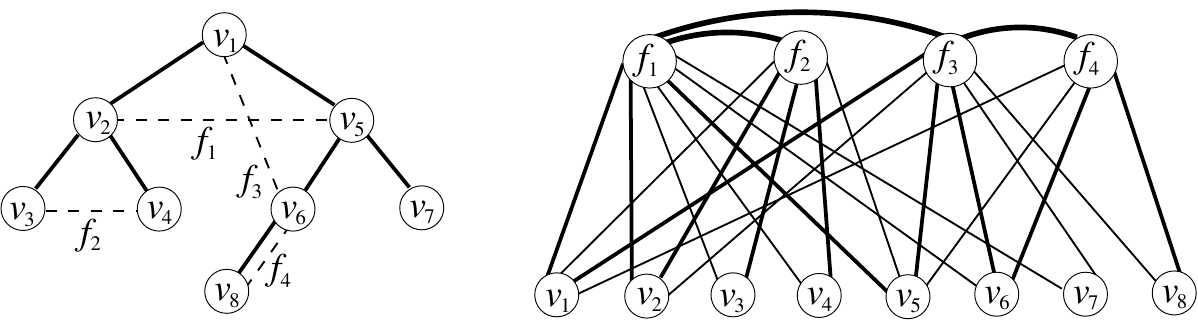}
\caption{Illustration to Definition~\ref{d:con-dom}. Here $E=\{f_1,f_2,f_3,f_4\}$ 
and $\{f_1,f_3\}$ is a unique optimal solution.
In the connectivity-domination graph on the right, the edges in $I$ are shown by bold arcs, 
and the edges in $D$ are shown by straight lines; 
the ``membership edges'' shown by bold lines encode that $v \in T_f$, 
while the other edges in $D$ shown by thin lines encode that $v$ is dominated but does not belong to $T_f$.}
\label{f:red}
\end{figure}

Note that an edge in $fg \in I$ encodes that $T_f$ and $T_g$ have a node in common, 
while an edge $ev \in D$ encodes that $v$ is dominated by $T_e$ 
(belongs to $T_e$ or is connected by an edge of $G$ to some node in $T_e$).
From Lemma~\ref{l:sv} we have:

\begin{corollary} \label{c:H}
$F \subs E$ is a feasible solution to a {\dt} instance if and only if 
in the connectivity-domination graph $\hat{G}$ the following holds:  \\
{\em (i)}  $\hat{G}[F]$ is connected; {\em (ii)} $F$ dominates $V$.
\end{corollary}

Our goal is to give an $O(\log n)$ approximation algorithm for the problem of finding min-size $F \subs E$ as in Corollary~\ref{c:H}.
Note that in this problem $V, E$ are both subsets of nodes of $\hat{G}$.
This is a particular case of the following problem.


\begin{center} \fbox{\begin{minipage}{0.97\textwidth}
\underline{{\su} {\st} {\c} {\ds}} \\
{\em Input:} \ \ A graph $\h{G}=(\h{V},\h{E})$ and a partition $Q,R$ of $\h{V}$. \\
{\em Output:} A min-size $S \subs Q$ such that $\h{G}[S]$ is connected and $S$ dominates $R$.
\end{minipage}} \end{center}


In Section~\ref{s:hard} we observe that up to constants, the approximability  of this problem is the same as 
that of the {\sc Group Steiner Tree} problem, that admits ratio $O(\log^3 n)$ \cite{GKR}.
However, in some cases better ratios are possible. 
In the case $Q=\h{V}$ we get the (unweighted) {\st} {\c} {\ds} problem, that admits ratio $O(\log n)$ \cite{GK98}. 
We show ratio $O(\log n)$ when $\h{G}$ is the connectivity-domination graph, with $Q=E$ and $R=V$. 
In what follows, given a {\su} {\st} {\c} {\ds} instance, let $q$ be the least integer such that 
for every $v \in R$, any two neighbors of $v$ in $\h{G}$ are connected by a path in $\h{G}[Q]$ that has at most $q$ internal nodes.

\begin{lemma}
{\su} {\st} {\c} {\ds} admits approximation ratio $O(q \log n)$ if $Q,R$ partition $\h{V}$ and $R$ is an independent set in $G$.
\end{lemma}
\begin{proof}
We find an $O(\log n)$ approximate solution $S$ for the {\st} {\c} {\ds} instance (with $Q=\h{V}$) using the algorithm of \cite{GK98}.
If $S \subs Q$ then we are done. Else, let $T=(V_T,E_T)$ be a subtree of $\h{G}$ with node set $S$.
We may assume that $T$ has no leaf in $R$, otherwise such leaf can be removed from $S$ and from $T$.
Let $R_T=S \cap R$ and $Q_T=S \cap Q$.  
Since $R$ is an independent set in $\h{G}$, $Q_T$ dominates $R$, 
and the nodes in $R_T$ are used in $S$ just to connect between the nodes in $Q_T$.
Moreover, since $R_T$ is an independent set in $T$  
$$
\sum_{r \in R_T} (\deg_T(r)-1) \leq |E_T|-|R_T|=|S|-1-|R_T| =|Q_T|-1 \ .
$$
Let $r \in R_T$. Add a set of $\deg_T(r)-1$ dummy edges that form a tree on the neighbors of $r$ in $T$,
and then replace every dummy edge $uv$ by a path $P_{uv}$ in $\h{G}[Q]$ that has at most $q$ internal nodes.
Applying this on every $r \in R_T$ gives a connected graph in $\h{G}[Q]$ that contains the set $Q_T$ that dominates $R$,
and the number of nodes in this graph is at most
$$
|Q_T|+q \sum_{r \in R_T} (\deg_T(r)-1) \leq (q+1)(|Q _T|-1) \leq (q+1)|S| \ .
$$
Since $|S|$ is $O(\log n)$ times the optimum, the lemma follows.
\qed
\end{proof}

Note that in our case, when $\h{G}$ is the connectivity-domination graph, we have $Q=E$ and $R=V$. 
Then $Q,R$ partition $\h{V}$ and $R$ is an independent set in $G$, by the construction. 
The next lemma shows that $q$ is a small constant in our case. 

\begin{lemma}
$q\leq 2$ if $\h{G}$ is the connectivity-domination graph and $Q=E$, $R=V$.
\end{lemma}
\begin{proof}
Let $v \in R$ and let $Q_v$ be the set of neighbors of $v$ in $\h{G}$. 
For $e \in Q_v$ let $f_e$ be defined as follows:
\begin{itemize}
\item
If there exists some $f \in E$ (possibly $f=e$) such that $v \in T_f$ and $T_e,T_f$ have a node  in common,
then we say that {\bf $e$ is of type 1} and let $f_e=f$.
\item
If $f$ as above does not exist then we say that {\bf $e$ is of type 2} and let $f_e=e$.
\end{itemize}
For illustration, consider the example in Fig.~\ref{f:red}. 
\begin{itemize}
\item
Let $v=v_6$. Then $Q_v=\{f_1,f_3,f_4\}$. 
If $e=f_3$ or if $e=f_4$ then $v \in T_f$, hence $e$ is of type~1 and we may have $f_e=e$. 
If $e=f_1$ then $v \notin T_e$, but $e$ is still of type~1 since for $f=f_3$ we have $v \in T_f$ and $T_e,T_f$ have a node  in common.
\item
Let $v=v_7$. Then $Q_v=\{f_1,f_3\}$.  There is no $f \in E$ such that $v \in T_f$, hence both $f_1,f_3$ are of type~2.
\end{itemize}

Suppose that every $e \in Q_v$ is of type~1.
Then for every $e \in Q_v$ we have $f_e=e$ or $ef_e \in I$, and note that $v \in T_{f_e}$.
Thus for any $e_1,e_2 \in Q_v$, the sequence $e_1,f_{e_1},f_{e_2},e_2$ forms a path in $\h{G}[Q]$
with at most two internal nodes.

Suppose that there is $e \in Q_v$ of type~2.
Then $v \notin T_e$, and $v$ is dominated by $T_e$ via some edge $uv$ of $T$; otherwise $e$ is of type~1.
Let $T^u$ and $T^v$ be the two subtrees of $T \sem e$, where $u \in T^u$ and $v \in T^v$.
Note that no edge in $E$ connects $T^u$ and $T^v$;  otherwise $e$ is of type~1.
Hence $uv$ is a bridge of $G$. 
This implies $T^v$ consists of a single node $v$, as otherwise the instance has no feasible solution.
Consequently, every $e \in Q_v$ is of type~2 and $u \in T_e$ holds, hence $\h{G}[Q_v]$ is a clique, and the lemme follows.
\qed
\end{proof}

This concludes the proof of Theorem~\ref{t:main'}, and thus also the proof of Theorem~\ref{t:main}. 

\section{{\sc Connected Dominating Set} variants} \label{s:hard}

Here we make some observations about the approximability of several variants of the 
{\sc Connected Dominating Set} ({\sc CDS}) problem.
In all these variants we are given a graph $G=(V,E)$ and possibly edge-costs/node-weights, 
and seek a minimum cost/weight/size subtree $H=(V_H,E_H)$ of $G$ that satisfies a certain domination property. 
Recall that in the {\sc CDS} problem $V_H$ should dominate $V$. 
The additional variants we consider are as follows.

\medskip

\noindent \underline{\sc Steiner CDS}: 
$V_H$ dominates a given set of terminals $R \subs V$. \\

\noindent \underline{\sc Subset Steiner CDS}:
$V_H$ dominates $R$ and $V_H \subs Q$ for a partition $Q,R$ of $V$. \\

\noindent \underline{\sc Partial CDS}: 
$V_H$ dominates at least $k$ nodes. \\

We relate these problems to the {\sc Group Steiner Tree} ({\sc GST}) problem: 
given a graph $G=(V,E)$ and a collection ${\cal S}$ of groups (subsets) of $V$,
find a minimum edge-cost/node-weight/size subtree $H$ of $G$ that contains at least one node from every group.
When the input graph is a tree and there are $k$ groups, edge-costs {\sc GST} admits ratio $O(\log n \log k)$ \cite{GKR}, 
and this is essentially tight \cite{HK}.
For general graphs the edge-costs version admits ratio $O(\log^2 n \log k)$, 
using the result of \cite{GKR} for tree inputs and the \cite{FRT} probabilistic tree embedding. 
However, the best ratio known for the node-weighted {\sc GST} is the one that is derived from the more general 
{\sc Directed Steiner Tree} problem \cite{CCCD,Z-DST,HRZ,KP-DST} with $k$ terminals -- 
for any integer $1 \leq \ell \leq k$, ratio $O\left(\ell^3 k^{2/\ell}\right)$ in time 
$O\left(k^{2\ell}n^\ell\right)$.
  
As was observed in \cite{KPS}, several {\sc CDS} variants are particular cases of the corresponding {\sc GST} variants,
where for every relevant node $r$ we have a group $S_r$ of nodes that dominate $r$ in the input graph. 
Specifically, we have the following.

\begin{lemma} \label{l:SS<GST}
For edge-costs/node-weights, ratio $\al(n,k)$ for {\sc GST} with $n$ nodes and $k$ groups 
implies ratio $\al(|Q|,|R|)$ for {\sc Subset Steiner CDS},
and this is so also for the unit node weights versions of the problems.
\end{lemma}
\begin{proof}
Given a {\sc Subset Steiner CDS} instance $G=(V,E)$ with edge costs/node weights and $Q,R \subs V$, 
construct a {\sc GST} instance by introducing for every $r \in R$ a group $S_r$ of nodes in $Q$ that dominate $r$.
In all cases, we return an $\al$-approximation for the {\sc GST} instance on $G[Q]$, that has $|Q|$ nodes and $|R|$ groups.
\qed
\end{proof}

Earlier, Guha and Khuller \cite{GK98} showed that the inverse is also true for edge-costs {\sc CDS} and node-weighted 
{\sc Steiner CDS}; in \cite{GK98} the reduction was to the {\sc Set TSP} problem, that can be shown to have 
the same approximability as GST, up to a factor of 2. 
Note that already the edge-costs {\sc CDS} is {\sc GST} hard,
hence our ratio $\tilde{O}(\log^2 n)$ for unit edge costs {\sc $2$-Edge-Connected Dominating Set} 
is unlikely to be extended to arbitrary costs.

We now show that {\sc Subset Steiner CDS} with unit edge costs/node weights
is hard to approximate as {\sc GST} with general edge costs. 
In what follows, we will assume that ratio $\al(n)$ for a given problem is an increasing function of $n=|V|$.

\begin{theorem} \label{t:S>GST}
For any constant $\eps>0$, ratio $\al(n)$ for {\sc Subset Steiner CDS} with unit edge costs/node weights
implies ratio $\al(|E|(n+k)/\eps)+\eps$ for {\sc GST} with arbitrary edge costs.
\end{theorem}

Theorem~\ref{t:S>GST} is proved in the next two lemmas.
Note that combined with Lemmas \ref{l:SS<GST}, Theorem~\ref{t:S>GST} implies
that the approximability of  
{\sc Subset Steiner CDS} with unit edge costs/node weights 
is essentially the same as that of {\sc GST} with arbitrary edge costs, up to a constant factor.
Recall that we showed that particular instances of {\sc Subset Steiner CDS} with unit node weights admit ratio $O(\log n)$.
Theorem~\ref{t:S>GST} implies that we could not achieve this for general unit node weights instances.

The next lemma shows that for unit edge costs/node weights, 
{\sc Subset Steiner CDS} is not much easier than {\sc GST}.

\begin{lemma}
For unit edge costs/node weights, ratio $\al(n)$ for {\sc Subset Steiner CDS}  
implies ratio $\al(n+k)$ for {\sc GST} with $k$ groups.
\end{lemma}
\begin{proof}
For each one of the problems in the lemma, any inclusion minimal solution is a tree,
hence the unit edge costs case is equivalent to the unit node weights case;
this is so up to an additive $\pm 1$ term, which can be avoided by guessing an edge/node that belongs to some optimal solution.
So we will consider just the unit node weights case. 
Given a unit weight {\sc GST} instance $G=(V,E),{\cal S}$ construct a unit weight 
{\sc Subset Steiner CDS} instance $G'=(V',E'),Q,R$ as follows.
For each group $S \in {\cal S}$ add a new node $r_S$ connected to all nodes in $S$.
The set of nodes that should be dominated is $R=\{r_S: S \in {\cal S}\}$, and $Q=V$.
Any subtree $H$ of $G'[Q]$ is is also a subtree of $G$, 
and for any group $S \in {\cal S}$, $H$ contains a node from $S$ if and only if 
$H$ dominates $r_S$. Thus $H$ is a feasible {\sc GST} solution if and only if $H$ is a feasible {\sc Subset Steiner CDS} solution. 
\qed
\end{proof}

The next lemma shows that {\sc GST} with unit edge costs is not much easier than 
{\sc GST} with arbitrary edge costs.

\begin{lemma} \label{l:GST}
If {\sc GST} with unit edge costs admits ratio $\al(n)$  
then for any constant $\eps>0$ {\sc GST} (with arbitrary edge costs) admits ratio $\al(n|E|/\eps)+\eps$.
\end{lemma}
\begin{proof}
Let $G=(V,E),c$ be a {\sc GST} instance (with arbitrary edge costs).
Fix some optimal solution $H^*$. Let $M=\max_{e\in H^*} c(e)$ be the maximum cost of an edge in $H^*$.
Note that $c(H^*) \geq M$. 
Since there are $O(n^2)$ edges, we can guess $M$, and remove from $G$ all edges of cost greater than $M$.

Let $\mu=\frac{\eps M}{n}$. Define new costs by $c'(e)=\lfloor\frac{c(e)}{\mu}\rfloor$. 
Note that $\mu c'(e)\leq c(e)\leq\mu(c'(e)+1)$ for all $e$, thus for any edge set $J$ with $|J| \leq n$ 
$$
c(J) = \sum_{e \in J} c(e) \leq  \sum_{e \in J} \mu(c'(e)+1)=\mu c'(J)+\mu|J| \leq \mu c'(J)+\eps M \ .
$$ 
This implies that if $H$ is an $\al$-approximate solution w.r.t. the new costs $c'$ then
$$
c(H) \leq \mu c'(H)+\eps M \leq \mu \al c'(H^*)+\eps M \leq \al c(H^*) +\eps c(H^*)=(\al+\eps)c(H^*) \ .
$$
So ratio $\al$ w.r.t. costs $c'$ implies ratio $(\al+\eps)$ w.r.t. the original costs $c$.

The instance with costs $c'$ can be transformed into an equivalent instance with unit edge costs 
and at most $n|E|/\eps$ nodes by a folklore reduction that replaces every edge by a path of length 
equal to the cost of the edge. Note that $c'$ is integer valued and that 
$$
c'(e) \leq \frac{c(e)}{\mu} = \frac{c(e)\cdot n}{\epsilon M} \leq \f{n}{\eps}  \ .
$$
Thus the number of nodes in the obtained unit edge costs instance is bounded by $n+|E|(n/\eps-1) \leq |E|n/\eps$.

Given the {\sc GST} instance with costs $c'$, 
we contract every zero cost edge while updating the groups accordingly.
Then we replace every edge $e=uv$ by a $uv$-path of length $c'(e)$, thus obtaining an equivalent 
{\sc GST} instance with unit edge costs and at most $|E| n/\eps$ nodes. 
\qed
\end{proof}

Theorem~\ref{t:S>GST} follows from the last two lemmas.

Finally, we consider the {\sc Partial CDS} problem. 
For unit node weights the problem was shown to admit a logarithmic ratio in \cite{KPS}.
We show that in the case of arbitrary weights, the problem 
is not much easier than {\sc Subset Steiner CDS}, and thus also is not easier than {\sc GST}.

\begin{lemma} \label{l:partial}
For edge costs/node weights, ratio $\al(n)$ for {\sc Partial CDS} 
implies ratio $\al(n^2)$ for {\sc Subset Steiner CDS}.
\end{lemma} 
\begin{proof}
Let us consider the case of node-weights.
Given a {\sc Subset Steiner CDS} instance $G,w,Q,R$
construct a  {\sc Partial CDS} instance $G',w',k$ as follows.
The graph $G'$ is obtained from $G$ by adding $|Q|$ copies $R_1,\ldots,R_{|Q|}$ of $R$,
and for each $r \in R$ connecting each copy $r_i \in R_i$ of $r$ to all nodes in $Q$ that dominate $r$.
We let $w'(v)=w(v)$ if $v \in Q$ and $w'(v)=\infty$ otherwise, and we let $k=(|Q|+1)|R|$.
In the obtained {\sc Partial CDS} instance, a subset of $Q$ that does not dominate $R$,
dominates at most $(|R|-1)(|Q|+1)+|Q|=k-1$ nodes; 
hence any feasible solution of finite weight must dominate $R$.
The {\sc Partial CDS} instance has $|Q||R|+ |R|+|Q| \leq n^2$ nodes, and the node weights case follows.

In the case of edge costs $G',k$ are as in the case of node weights, 
and the cost of an edge $uv$ of $G'$ is $c'(uv)=c(uv)$ if $u,v \in Q$ and $c'(uv)=\infty$ otherwise. 
The rest of the proof is the same as in the case of node-weights.
\qed
\end{proof}


\end{document}